\documentclass[onecolumn, 11pt]{article}
\usepackage{amsmath,amssymb,amsthm}
\usepackage{bm}
\usepackage{url}
\usepackage{tikz}

\allowdisplaybreaks[4]

\newtheorem{theorem}{Theorem}
\newtheorem{lemma}[theorem]{Lemma}
\newtheorem{corollary}[theorem]{Corollary}
\theoremstyle{definition}
\newtheorem{definition}[theorem]{Definition}

\theoremstyle{remark}

\title{Peeling Algorithm on Random Hypergraphs with Superlinear Number of Hyperedges}
\author{Ryuhei~Mori~~and~~Osamu~Watanabe}
\date{
Tokyo Institute of Technology, Tokyo, Japan\\
email: mori@is.titech.ac.jp, watanabe@is.titech.ac.jp}

\begin{document}
\maketitle

\begin{abstract}
When we try to solve a system of linear equations, we can consider a simple iterative algorithm
in which an equation including only one variable is chosen at each step, and the variable is fixed to the value satisfying the equation.
The dynamics of this algorithm is captured by the peeling algorithm.
Analyses of the peeling algorithm on random hypergraphs are required for many problems, e.g.,
the decoding threshold of low-density parity check codes, the inverting threshold of Goldreich's pseudorandom generator,
the load threshold of cuckoo hashing, etc.
In this work, we deal with random hypergraphs including superlinear number of hyperedges, and derive the tight threshold for the succeeding of the peeling algorithm.
For the analysis, Wormald's method of differential equations, which is commonly used for analyses of the peeling algorithm on
random hypergraph with linear number of hyperedges, cannot be used due to
the superlinear number of hyperedges.
A new method called the evolution of the moment generating function is proposed in this work.
\end{abstract}

\section{Introduction}
The peeling algorithm is a simple message passing algorithm on hypergraph, which has been used for
analysis of many practical problems, e.g.,
the decoding of low-density parity-check codes~\cite{Luby:1997:PLC:258533.258573},
 the satisfiability and clustering phase transition of random $k$-XORSAT~\cite{ibrahimi2011set},
 load threshold of cuckoo hashing~\cite{dietzfelbinger2010tight},
invertible Bloom lookup table~\cite{mitzenmacher2013simple}, etc.
The peeling algorithm works on a bipartite graph representation of a hypergraph consisting of vertex nodes and hyperedge nodes.
In the $d$-peeling algorithm, hyperedge nodes of degree at most $d-1$ are iteratively removed.
In this work, we consider the peeling algorithm on randomly generated $k$-uniform hypergraph with superlinear number of hyperedges where sublinear number of vertices are initially removed.
Problems of this type were considered in~\cite{watanabe2013mp}, \cite{coja2012propagation}.
The results of this paper are useful for analyses of message passing algorithm for planted MAX-$k$-LIN and planted uniquely extendible constraints satisfaction problems~\cite{watanabe2013mp}, \cite{connamacher2012exact}
 and the inverting algorithm for Goldreich's generator~\cite{odonnell2014goldreich}.
For analyses of the peeling algorithm, two methods have been used in the previous works: the density evolution~\cite{mct} and Wormald's method of differential equation~\cite{Luby:1997:PLC:258533.258573}, \cite{wormald1995differential}.
The density evolution is not available on our setting since the hypergraph is not locally tree due to the superlinear number of hyperedges.
Wormald's method is also not available since the numbers of hyperedges with particular degrees in the peeling process are highly biased, e.g.,
the number of degree-1 hyperedge nodes is sublinear while the number of degree-3 hyperedge nodes is superlinear.
The analysis in this work is founded on the Markov chain of the number of hyperedge nodes which has been also used in Wormald's method~\cite{Achlioptas2001159}, \cite{connamacher2012exact}.
We analyze the peeling algorithm by introducing the evolution of the moment generating function, which gives the precise analysis of the behavior of the peeling algorithm.

\section{Main results}
In this work, we consider randomly generated hypergraphs.
\begin{definition}[Random hypergraph]
A random hypergraph $G_k(n, m(n), \ell(n))$ is defined by the following generating process.
First, $k$-uniform hypergraph is generated by choosing $m(n)$ hyperedges independently and uniformly 
from all of the $\binom{n}{k}$ size-$k$ subsets of $n$ vertices.
Second, $\ell(n)$ randomly chosen vertices are removed from the $k$-uniform hypergraph.
Equivalently, we can assume that the $\ell(n)$ vertices with smallest indices are removed.
\end{definition}
In this paper, we always assume $\ell(n)\in\omega(1)\,\cap\,o(n)$.
For a given hypergraph generated randomly as above, the $d$-peeling algorithm, that we consider in this paper, is an algorithm
iteratively removing hyperedge nodes of degree at most $d-1$ until no
such node exists (see Definition~\ref{def:peeling} for the formal definition).
The behavior of the $k$-peeling algorithm on
$G_k(n,m(n),\ell(n))$ is essentially
determined by the connectivity of the random $k$-uniform hypergraph $G_k(n,m(n),0)$
since 
vertices of $G_k(n,m(n),\ell(n))$ removed by the $k$-peeling algorithm are
those which were connected to some of the $\ell(n)$ vertices removed from the random $k$-uniform hypergraph.
Hence, the asymptotic behavior of the $k$-peeling algorithm on $G_k(n,m(n),\ell(n))$ is 
derived from the phase transition phenomenon of the connectivity of $G_k(n,m(n),0)$~\cite{schmidt1985component} (See also Appendix~\ref{apx:r2}).
In this work, we show the phase transition phenomenon of the $d$-peeling algorithm for $d\le k-1$.
Let the threshold constant  be
$\mu_\mathrm{c}(k,r):=\binom{k}{r}^{-1}\frac{(r-2)^{r-2}}{r(r-1)^{r-1}}$.
The followings are the main results of this paper.
\begin{theorem}\label{thm:main00}
Let $m(n)=\mu \frac{n^{r-1}}{\ell(n)^{r-2}}$ for arbitrary constant $\mu > \mu_{\mathrm{c}}(k,r)$.
Then, the $(k-r+2)$-peeling algorithm removes $n-o(n)$ vertices of $G_k(n,m(n),\ell(n))$
with high probability for $r\in\{3,\dotsc,k\}$.
In addition, if $m(n)=\omega(n\log n)$, i.e., $\ell(n)=o(n/(\log n)^{1/(r-2)})$, the $(k-r+2)$-peeling algorithm removes all vertices of $G_k(n,m(n),\ell(n))$ with high probability.
\end{theorem}
\begin{theorem}\label{thm:main11}
Let $m(n)=\mu \frac{n^{r-1}}{\ell(n)^{r-2}}$ for arbitrary constant $\mu < \mu_{\mathrm{c}}(k,r)$.
Then, the $(k-r+2)$-peeling algorithm removes only $\Theta(\ell(n))$ vertices of $G_k(n,m(n),\ell(n))$
with high probability for $r\in\{3,\dotsc,k\}$.
\end{theorem}
The above results show that $\mu_{\mathrm{c}}(k,r)$ is the sharp threshold constant for
 the behavior of the peeling algorithm.
Furthermore, upper bounds of the rate of the large deviation and the number of removed vertices below the threshold are also obtained in this paper.

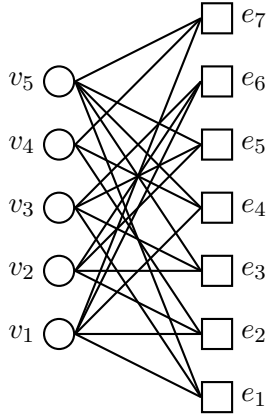
\begin{figure}
\begin{center}
\begin{tikzpicture}
[yshift=20pt, rotate=90,scale=0.42, inner sep=0mm, C/.style={minimum size=4mm,circle,draw=black,thick},
S/.style={minimum size=4mm,rectangle,draw=black,thick}, label distance=1mm]
\node (1) at (0,5.0) [C,label=left:$v_1$] {};
\node (2) at (2,5.0) [C,label=left:$v_2$] {};
\node (3) at (4,5.0) [C,label=left:$v_3$] {};
\node (4) at (6,5.0) [C,label=left:$v_4$] {};
\node (5) at (8,5.0) [C,label=left:$v_5$] {};
\node (a1) at (-2,0) [S,label=right:${e_1}$] {};
\node (a2) at (0,0) [S,label=right:${e_2}$] {};
\node (a3) at (2,0) [S,label=right:${e_3}$] {};
\node (a4) at (4,0) [S,label=right:${e_4}$] {};
\node (a5) at (6,0) [S,label=right:${e_5}$] {};
\node (a6) at (8,0) [S,label=right:${e_6}$] {};
\node (a7) at (10,0) [S,label=right:${e_7}$] {};
\draw (a1.west) to (1.east) [thick];
\draw (a1.west) to (3.east) [thick];
\draw (a1.west) to (5.east) [thick];
\draw (a2.west) to (1.east) [thick];
\draw (a2.west) to (2.east) [thick];
\draw (a2.west) to (4.east) [thick];
\draw (a3.west) to (2.east) [thick];
\draw (a3.west) to (3.east) [thick];
\draw (a3.west) to (5.east) [thick];
\draw (a4.west) to (1.east) [thick];
\draw (a4.west) to (4.east) [thick];
\draw (a4.west) to (5.east) [thick];
\draw (a5.west) to (2.east) [thick];
\draw (a5.west) to (3.east) [thick];
\draw (a5.west) to (5.east) [thick];
\draw (a6.west) to (1.east) [thick];
\draw (a6.west) to (2.east) [thick];
\draw (a6.west) to (3.east) [thick];
\draw (a7.west) to (1.east) [thick];
\draw (a7.west) to (4.east) [thick];
\draw (a7.west) to (5.east) [thick];
\end{tikzpicture}
\end{center}
\caption{A bipartite graph representation of a 3-uniform hypergraph.}
\label{fig:hypergraph}
\end{figure}

\section{Bipartite graph representation of hypergraphs, peeling algorithm and stopping sets}
In this work, a hypergraph is represented by a bipartite graph.
The bipartite graph representation consists of two types of nodes ``vertex nodes'' and ``hyperedge nodes''
each of which corresponds to a vertex and a hyperedge in the hypergraph, respectively.
A vertex node $v$ and a hyperedge node $e$ are connected by an edge in the bipartite graph representation if and only if
 the vertex corresponding to $v$ is a member of the hyperedge corresponding to $e$ in the hypergraph.
The example of the bipartite graph representation is shown in Fig.~\ref{fig:hypergraph}.
The set of vertex nodes and the set of hyperedge nodes are denoted by $V$ and $E$, respectively.
The neighborhoods of vertex node $v\in V$ and the neighborhoods of hyperedge node $e\in E$ are denoted by $\partial v\subseteq E$ and $\partial e\subseteq V$, respectively.
\begin{definition}[Peeling algorithm for a bipartite graph]\label{def:peeling}
For $d\in\{2,3,\dotsc,k\}$, the $d$-peeling algorithm for a bipartite graph is defined as follows.
If there is a hyperedge node $e\in E$ of degree at most $d-1$,
then the hyperedge node $e$ and all of the at most $d-1$ vertex nodes connected to the hyperedge node $e$
are removed from the bipartite graph.
This process is iterated until there is no hyperedge node of degree at most $d-1$.
\end{definition}
Note that on similar settings, the peeling algorithm was analyzed for $k=3$ and $d=2$ in~\cite{coja2012propagation}, \cite{watanabe2013mp}.
The peeling algorithm stops if and only if the current set of variables forms a structure called a stopping set.
\begin{definition}[Stopping set~{\cite{di2002finite}}]
For $d\ge 2$, a subset $S\subseteq V$ is called a $d$-stopping set if $|\partial e\cap S|\in \{0,d,d+1,\dotsc,k\}$ for all hyperedges $e\in
E$.
\end{definition}
It is obvious that the $d$-peeling algorithm terminates at the largest $d$-stopping set.
Hence, it is sufficient to analyze the existence of non-empty $d$-stopping sets for analyzing the $d$-peeling algorithm.
%
We classify non-empty $d$-stopping sets to three classes according to their size; 
$\alpha$-small $d$-stopping sets whose size is at least 1 and at most $\lceil \alpha n\rceil$,
$\alpha$-linear $d$-stopping sets whose size is at least $\lceil \alpha n\rceil+1$ and at most $\lfloor(1-\alpha)n\rfloor$
and $\alpha$-large $d$-stopping set whose size is at least $\lfloor (1-\alpha)n\rfloor +1$ for some fixed $\alpha\in(0,1/2)$.
%

\section{Analysis of stopping sets}
In this section, we show the results of analysis of existence of stopping sets which reveals the behavior of the peeling algorithm.
As shown in Appendix~\ref{apdx:linear}, by the standard analysis using Markov's inequality and expected number of stopping sets,
 it is easy to show that there is no $\alpha$-linear $d$-stopping set if the number $m(n)$ of hyperedges is superlinear.
\begin{lemma}[Linear-size stopping sets]\label{lem:linear}
For any $d\in\{2,3,\dotsc,k\}$ and $\alpha\in(0,1/2)$, there exists $\beta >0$ such that
$G_k(n,\beta n, 0)$ does not have $\alpha$-linear $d$-stopping set
with probability $1-\exp\{O(n)\}$.
\end{lemma}

Similarly, it is also shown in Appendix~\ref{apdx:small} that there is no $\alpha$-small stopping set if $m(n)=\omega(n\log n)$.
\begin{lemma}[Threshold for small stopping sets]\label{lem:small}
For any $d\in\{2,3,\dotsc,k\}$ and $\alpha\in(0,1/2)$,
$G_k(n, \mu n\log n, 0)$ does not have $\alpha$-small $d$-stopping set
with probability $1-O(n^{-\delta})$
for any $\mu>1/k$
and $\delta\in(0,\mu k-1)$.
\end{lemma}
Conversely, if $m(n)=\mu n\log n$ for $\mu<1/k$, from the theory of the coupon collector's problem,
with high probability there exists a vertex node which is not connected to any hyperedge node.
Therefore, there exists a $d$-stopping set of size 1 with high probability.
Hence, the constant $1/k$, which appears as a coefficient of $n\log n$, is the sharp threshold 
for the existence of small stopping sets.
While the above two Lemmas are obtained by Markov's inequality and analysis of expected number of $2$-stopping sets,
the analysis of $\alpha$-large $d$-stopping sets requires  more involved analysis of dynamics of the $d$-peeling algorithm.
Recall $\ell(n)\in\omega(1)\cap o(n)$.
The followings results on large stopping sets are shown in the next section.

\begin{theorem}\label{thm:ML}
Fix $r\ge 3$.
For any constant $\mu<\mu_\mathrm{c}(k,r)$,
$G_k(n,\mu \frac{n^{r-1}}{\ell(n)^{r-2}}, \ell(n))$ has
$(k-r+2)$-stopping set of size larger than $n-(1+\tau) \ell(n)$ with probability at least $1-p(n, \mu, \tau)$
for any $\tau>\tau^*$ where $\tau^*\in(0,1/(r-2))$ is the unique solution in $(0,1/(r-2))$ of
\begin{equation*}
\mu = \frac1{\binom{k}{r}}\frac{\tau^*}{r(1+\tau^*)^{r-1}}.
\end{equation*}
Here, the probability $p(n,\mu,\tau)$ is 
\begin{equation*}
 \exp\left\{\inf_{\lambda>0,\tau'\in(\tau^*,\tau)}\left\{ \varphi_{k,r}(\mu,\lambda,\tau')\right\} \ell(n)  + O(\max\{\ell(n)^2/n,1\})\right\}
\end{equation*}
where
\begin{equation}\label{eq:Ek}
\varphi_{k,r}(\mu,\lambda,\tau):= \mu \left(\exp\{(k-r+1)\lambda\}-1\right)\binom{k}{r-1}(1+\tau)^{r-1}- \lambda\tau.
\end{equation}
\end{theorem}

\begin{theorem}\label{thm:MU}
Fix $r\ge 3$.
For any $\alpha\in(0,1/2)$ and
for any constant $\mu>\mu_\mathrm{c}(k,r)$,
$G_k(n,\mu \frac{n^{r-1}}{\ell(n)^{r-2}},\ell(n))$ does not have
$\alpha$-large $(k-r+2)$-stopping set with probability at least
\begin{align*}
&1-\exp\left\{\sup_{\tau >0}\inf_{\lambda<0}\allowbreak \left\{\varphi_{k,r}(\mu,\lambda,\tau)\right\} \ell(n) + O(\max\{\ell(n)^2/n,\,\log \ell(n)\})\right\}.
\end{align*}
Here, it holds
\begin{equation*}
\exp\left\{\sup_{\tau >0}\inf_{\lambda<0}\allowbreak \left\{\varphi_{k,r}(\mu,\lambda,\tau)\right\}  \right\}
=\rho^{\frac{1-(r-2)\tau}{r-1}}
\end{equation*}
where $(\rho,\tau)$ is the solution of
\begin{align}
\label{eq:rho}
\rho &= \exp\left\{\mu \binom{k}{r-2}(1+\tau)^{r-2}(k-r+2) \left(\rho^{k-r+1}-1\right) \right\}\\
\mu\rho^{k-r+1} &= \frac{\tau}{\binom{k}{r}r(1+\tau)^{r-1}}.\label{eq:tau}
\end{align}
\end{theorem}
Note that for $\mu>(k(k-1))^{-1}$, $G_k(n, \mu n, 0)$ has a giant component whose size is concentrated around $(1-\rho)n$ where
$\rho$ satisfies~\eqref{eq:rho} for $r=2$~\cite{RSA:RSA20160}.
Hence, the equations~\eqref{eq:rho} and \eqref{eq:tau} may give the generalized concept of ``size of giant component'' (See also Appendix~\ref{apx:r2}).

\section{Evolution of the number of hyperedges in the peeling algorithm}\label{sec:evo}
\subsection{The Markov chain}
In this section, we analyze the numbers of hyperedges at each step of the iterations of the $(k-r+2)$-peeling algorithm
on $G_k(n,m(n),\ell(n))$.
In this section, we deal with arbitrary fixed $r\ge 2$.
For the analysis, we assume that only one hyperedge node $e\in E$ of degree at most $k-r+1$ is chosen in each step
and that only one of the vertex node connected to the hyperedge node $e$ is removed
from the bipartite graph.
Note that the scheduling of the peeling algorithm does not affect to the remaining graph after the termination of the peeling algorithm.
Let $C_j(t)$ be a random variable corresponding to the number of hyperedge nodes of degree $j$ after $t$ iterations
for $j\in[k]:=\{1,2,\dotsc,k\}$ and $t\in\{1,2,\dotsc\}$.
Obviously, $[C_0(0),\dotsc,C_k(0)]$ obeys the multinomial distribution
$\mathrm{Multinom}(m(n), p_0(n), p_1(n),\dotsc, p_k(n))$ where
\begin{align*}
p_j(n) := \frac{\binom{n-\ell(n)}{j}\binom{\ell(n)}{k-j}}{\binom{n}{k}}=
\binom{k}{j}\frac{\ell(n)^{k-j}}{n^{k-j}} + O\left(\frac{\ell(n)^{k-j+1}}{n^{k-j+1}}\right).
\end{align*}
Let $[B_{1}(t),B_2(t),\dotsc, B_{k-r+1}(t)]$ be a 0-1 random vector of weight 1 where $B_j(t)=1$ if a hyperedge node of degree $j$ is chosen at $(t+1)$-th
iteration
and $B_j(t)=0$ otherwise.
We assume that a hyperedge node is chosen uniformly from all hyperedge nodes of degree at most $k-r+1$.
Hence,
\begin{equation*}
\Pr\left(B_j(t)=1\mid [C_0(t),\dotsc,C_k(t)]\right) = \frac{C_j(t)}{\sum_{j'=1}^{k-r+1} C_{j'}(t)}
\end{equation*}
if $\sum_{j=1}^{k-r+1}C_j(t)\ge 1$.
Note that the distribution of $[B_{1}(t),B_2(t),\dotsc, B_{k-r+1}(t)]$ is not used in the following analysis.
Let $N(t):=n-\ell(n)-t$ be the number of remaining vertex nodes after $t$ iterations when the iterations continues until the $t$-th step.
The set of random variables $([C_0(t),\dotsc,C_k(t)])_{t=0,1,\dotsc,N(0)}$
is a Markov chain satisfying $[C_0(t+1),\dotsc,C_k(t+1)]=[C_0(t),\dotsc,C_k(t)]$ if $\sum_{j=1}^{k-r+1}C_j(t)=0$ and
\begin{equation}
\begin{split}
C_k(t+1)&= C_k(t) - R_k(t)\\
C_{j}(t+1)&= C_{j}(t) - R_{j}(t) + R_{j+1}(t),\hspace{2em} \text{for}\hspace{1em} j=1,2,\dotsc,k-1\\
C_0(t+1)&= C_0(t) + R_1(t)
\end{split}
\label{eq:Markov}
\end{equation}
if $\sum_{j=1}^{k-r+1}C_j(t)\ge 1$
where $R_1(t),\dotsc, R_k(t)$ are independent random variables conditioned on $[C_0(t),\dotsc,C_k(t)]$ and $[B_1(t),\dotsc,B_{k-r+1}(t)]$
obeying
\begin{align*}
R_j(t) &\sim \mathrm{Binom}\left(C_j(t), \frac{j}{N(t)}\right),\hspace{2em} \text{for}\hspace{1em} j=k-r+2, k-r+3,\dotsc,k\\
R_j(t) &\sim B_j(t)+\mathrm{Binom}\left(C_j(t)-B_j(t), \frac{j}{N(t)}\right),\hspace{2em} \text{for}\hspace{1em} j=1,2,\dotsc,k-r+1.
\end{align*}
Let $E_1^{k-r+1}(t):= \sum_{j=1}^{k-r+1} j C_j(t)$ be the number of edges connected to hyperedge nodes of degree at most $k-r+1$.
Then, the probability that $G_k(n,m(n),\ell(n))$ does not have $(k-r+2)$-stopping set of size larger than $n-\ell(n)-t$ is
exactly equal to
\begin{equation}
\Pr\left(E_1^{k-r+1}(0)\ge 1,E_1^{k-r+1}(1)\ge 1,\dotsc,E_1^{k-r+1}(t-1)\ge 1\right).
\label{eq:ps}
\end{equation}
For proving Theorems~\ref{thm:ML} and \ref{thm:MU}, we analyze the probability~\eqref{eq:ps}.
Similar analysis was considered in~\cite{Luby:1997:PLC:258533.258573},  \cite{Achlioptas2001159}, \cite{connamacher2012exact},
in which the number of hyperedge nodes $m(n)$ is proportional to $n$.
In that case, one can use Wormald's theorem, which gives differential equations describing the behavior of the Markov chain at
the limit $n\to\infty$~\cite{wormald1995differential}.
In this paper, $m(n)$ is not necessarily proportional to $n$.
Therefore, different techniques are required.

\subsection{Dominating Markov chain}
In this section, we prove Theorem~\ref{thm:ML}.
For the Markov chain~\eqref{eq:Markov}, it holds
\begin{equation}
\begin{split}
C_k(t+1)&= C_k(t) - R_k(t)\\
C_{j}(t+1)&= C_{j}(t) - R_{j}(t) + R_{j+1}(t),\hspace{2em} \text{for}\hspace{1em} j=k-r+2,k-r+3,\dotsc,k-1\\
E_1^{k-r+1}(t+1)&= E_1^{k-r+1}(t) - \sum_{j=1}^{k-r+1}R_j(t) + (k-r+1) R_{k-r+2}(t)
\end{split}
\label{eq:MarkovE}
\end{equation}
if $E_1^{k-r+1}(t)\ge 1$.
For upper bounding~\eqref{eq:ps}, we consider the dominating Markov chain
$([\overline{E}_1^{k-r+1}(t),\allowbreak\overline{C}_{k-r+2}(t),\dotsc,\allowbreak\overline{C}_k(t)])_{t=0,1,\dotsc,N(0)}$
 which satisfies
$\overline{E}_1^{k-r+1}(0)=E_1^{k-r+1}(0)$,
$\overline{C}_j(0)=C_j(0)$ for $j=k-r+2,\dotsc,k$ and
\begin{equation}
\begin{split}
\overline{C}_k(t+1)&= \overline{C}_k(t)\\
\overline{C}_{j}(t+1)&= \overline{C}_{j}(t) + \overline{R}_{j+1}(t),\hspace{2em} \text{for}\hspace{1em} j=k-r+2,k-r+3,\dotsc,k-1\\
\overline{E}_1^{k-r+1}(t+1)&= \overline{E}_1^{k-r+1}(t) - 1 + (k-r+1)\overline{R}_{k-r+2}(t)\\
\end{split}
\label{eq:MarkovU}
\end{equation}
where
\begin{align*}
\overline{R}_j(t) &\sim \mathrm{Binom}\left(\overline{C}_j(t), \frac{j}{N(t)}\right),\hspace{2em} \text{for}\hspace{1em} j=k-r+2,k-r+3,\dotsc,k.
\end{align*}
The dominating Markov chain does not have the conditioning $\overline{E}_1^{k-r+1}(t)\ge 1$ which appears in~\eqref{eq:MarkovE}.
Hence, it is easier to analyze the dominating Markov chain~\eqref{eq:MarkovU} than to analyze the original Markov chain~\eqref{eq:MarkovE}.
Obviously,~\eqref{eq:ps} is upper bounded by
\begin{equation}
\Pr\left(\overline{E}_1^{k-r+1}(0)\ge 1,\overline{E}_1^{k-r+1}(1)\ge 1,\dotsc,\overline{E}_1^{k-r+1}(t-1)\ge 1\right).
\label{eq:psU}
\end{equation}
While it is easy to derive and analyze recurrence equations of the expectations of the dominating Markov chain~\eqref{eq:MarkovU},
we will derive and analyze recurrence equation of the moment generating function of the dominating Markov chain~\eqref{eq:MarkovU}
for precise analysis.
By the analysis of the moment generating function in Section~\ref{sec:gen},
asymptotic behavior of the moment generating function of $\overline{E}_1^{k-r+1}(t)$ can be derived 
for $t=\Theta(\ell(n))$.

\begin{theorem}[Moment generating function of $\overline{E}_1^{k-r+1}(t)$]
\label{thm:UG}
Assume $m(n)=\mu\frac{n^{r-1}}{\ell(n)^{r-2}}$ for arbitrary constant $\mu$ and $\ell(n)\in \omega(1)\cap o(n)$.
Then, for any constants $\tau>0$ and $\lambda$, it holds
$\mathbb{E}[\exp\{\lambda\overline{E}_1^{k-r+1}(\lfloor\tau \ell(n)\rfloor)\}]=
\exp\{ \varphi_{k,r}(\mu,\lambda,\tau) \ell(n)\allowbreak+ O(\max\{1,\ell(n)^2/n\})\}$
where $\varphi_{k,r}(\mu,\lambda,\tau)$ is defined in~\eqref{eq:Ek}.
\end{theorem}
The proof of Theorem~\ref{thm:UG} is shown in Section~\ref{sec:gen}.
Now, Theorem~\ref{thm:ML} can be proved by using Theorem~\ref{thm:UG} and the Chernoff bound.

\begin{proof}[Proof of Theorem~\ref{thm:ML}]
From the Chernoff bound and Theorem~\ref{thm:UG}, one obtains an inequality
\begin{align*}
&\Pr\left(\overline{E}_1^{k-r+1}(\lfloor\tau \ell(n)\rfloor) \ge 1\right)
\le\Pr\left(\overline{E}_1^{k-r+1}(\lfloor\tau \ell(n)\rfloor) \ge 0\right)\\
&\quad\le\mathbb{E}[\exp\{\lambda\overline{E}_1^{k-r+1}(\lfloor\tau \ell(n)\rfloor)\}]
=\exp\{ \varphi_{k,r}(\mu,\lambda,\tau) \ell(n)+ O(\max\{1,\ell(n)^2/n\})\}
\end{align*}
for any constants $\tau\ge 0$ and $\lambda\ge 0$.
It holds
\begin{align*}
\frac{\partial \varphi_{k,r}(\mu,\lambda,\tau)}{\partial\lambda}
&=\mu \exp\{(k-r+1)\lambda\}r\binom{k}{r}(1+\tau)^{r-1}- \tau.
\end{align*}
If
\begin{equation}\label{eq:cond}
\left.\frac{\partial \varphi_{k,r}(\mu,\lambda,\tau)}{\partial\lambda}\right|_{\lambda=0}=
\mu r\binom{k}{r}(1+\tau)^{r-1}- \tau < 0
\end{equation}
then $\varphi_{k,r}(\mu,\lambda,\tau)$ is negative for sufficiently small $\lambda>0$ since $\varphi_{k,r}(\mu,0,\tau)=0$.
The condition~\eqref{eq:cond} is satisfied for some $\tau>0$ when
\begin{equation}\label{eq:sup}
\mu
< \frac1{r\binom{k}{r}}\sup_{\tau> 0} \frac{\tau}{(1+\tau)^{r-1}}.
\end{equation}
When $r\ge 3$, the supremum is taken at $\tau=1/(r-2)$, and hence the condition~\eqref{eq:sup} is equivalent to $\mu < \frac{(r-2)^{r-2}}{\binom{k}{r}r(r-1)^{r-1}}=\mu_{\mathrm{c}}(k,r)$.
When $\mu<\mu_\mathrm{c}(k,r)$, the inequality~\eqref{eq:cond} is satisfied for any $\tau\in(\tau^*,1/(r-2)]$.
That means that there exists $(k-r+2)$-stopping set of size at least $n-\lfloor(1+\tau)\ell(n)\rfloor$ with high probability
for any $\tau\in(\tau^*,1/(r-2)]$.
By optimizing $\tau$ and $\lambda$, one obtains Theorem~\ref{thm:ML}.
\end{proof}

\subsection{Dominated Markov chain}\label{subsec:lower}
In this section, we prove Theorem~\ref{thm:MU}.
We can use the same argument as Lemma~\ref{lem:2large} in Appendix~\ref{apx:r2} for the $(k-r+2)$-peeling algorithm.
For $m(n)=\mu\frac{n^{r-1}}{\ell(n)^{r-2}}$, it holds
\begin{equation*}
\mathbb{E}[C_{k-r+2}(0)]
=m(n) p_{k-r+2}(n) = \mu \binom{k}{r-2} n + O(\ell(n)).
\end{equation*}
Let as assume that there are 
$\mathbb{E}[C_{k-r+2}(0)]$ number of hyperedge nodes of degree $k-r+2$ with high probability.
In that case, if $\mu > [r(r-1)\binom{k}{r}]^{-1}$,
it holds
$\mathbb{E}[C_{k-r+2}(0)]>([(k-r+2)(k-r+1)]^{-1} + \delta)n$ for sufficiently small $\delta>0$.
Then, from the argument in the proof of Lemma~\ref{lem:2large}, linearly many vertex nodes are removed
 by the $(k-r+2)$-peeling algorithm with high probability.
However, $[r(r-1)\binom{k}{r}]^{-1}$ is strictly larger than $\mu_\mathrm{c}(k,r)$ for $r\ge 3$, and hence is not the sharp threshold.

In the following, we will show that if $\mu>\mu_\mathrm{c}(k,r)$,
for any $\eta>0$ there exists $\tau>0$ such that
\begin{equation}\label{eq:psl}
\Pr\left(E_1^{k-r+1}(0)\ge 1,\dotsc, E_1^{k-r+1}(\lfloor\tau \ell(n)\rfloor-1)\ge 1, E_1^{k-r+1}(\lfloor\tau \ell(n)\rfloor)\ge \eta \ell(n)\right)=1-o(1)
\end{equation}
and that if $\mu>\mu_\mathrm{c}(k,r)$,
there exists sufficiently small $\epsilon >0$ such that for any $\tau\ge 1/(r-2)$,
\begin{equation}\label{eq:p2l}
\Pr\left(C_{k-r+2}(\lfloor\tau \ell(n)\rfloor)>([(k-r+2)(k-r+1)]^{-1}+\epsilon) n\right) = 1-o(1).
\end{equation}
If the iteration of the peeling algorithm continues until $\lfloor \tau \ell(n)\rfloor$ steps and if
$E_1^{k-r+1}(\lfloor\tau \ell(n)\rfloor) \ge \eta \ell(n)$ and 
$C_{k-r+2}(\lfloor\tau \ell(n)\rfloor)>([(k-r+2)(k-r+1)]^{-1}+\epsilon) n$
hold,
then from the argument in the proof of Lemma~\ref{lem:2large}, 
$C_{k-r+2}(\lfloor\tau \ell(n)\rfloor)$ number of $(k-r+2)$-uniform hyperedges generate a giant component of size $(1-\rho)n$ for some $\rho\in(0,1)$ with high probability.
In that case, the peeling algorithm removes linearly many vertex nodes
with probability at least $1-\rho^{\eta \ell(n)}$.
Furthermore, from Lemma~\ref{lem:linear},
if $m(n)=\omega(n)$, there is no stopping set of linear size with high probability.
The above argument implies that~\eqref{eq:psl} and \eqref{eq:p2l}
give the proof of Theorem~\ref{thm:MU} except for the bound of the probability.

For lower bounding the probabilities in~\eqref{eq:psl} and~\eqref{eq:p2l}, we consider a dominated Markov chain
$([\underline{E}_1^{k-r+1}(t),\underline{C}_{k-r+2}(t),\allowbreak\dotsc,\underline{C}_k(t)])_{t=0,1,\dotsc,N(0)}$
 which satisfies
$\underline{E}_1^{k-r+1}(0)=\sum_{j=1}^{k-r+1}jC_j(0)$,
$\underline{C}_j(0)=C_j(0)$ for $j=k-r+2,\dotsc,k$ and
\begin{equation}
\begin{split}
\underline{C}_k(t+1)&= \underline{C}_k(t)-\underline{R}_k(t)\\
\underline{C}_{j}(t+1)&= \underline{C}_{j}(t)-\underline{R}_{j}(t) + \underline{R}_{j+1}(t),
\hspace{2em} \text{for}\hspace{1em} j=k-r+2,k-r+3,\dotsc,k-1\\
\underline{E}_1^{k-r+1}(t+1)&= \underline{E}_1^{k-r+1}(t) - 1 -\underline{R}_1^{k-r+1}(t) + (k-r+1)\underline{R}_{k-r+2}(t)\\
\end{split}
\label{eq:MarkovL}
\end{equation}
where
\begin{align*}
\underline{R}_j(t) &\sim \mathrm{Binom}\left(\underline{C}_j(t), \frac{j}{N(t)}\right),\hspace{2em} \text{for}\hspace{1em} j=k-r+2,k-r+3,\dotsc,k\\
\underline{R}_1^{k-r+1}(t) &\sim \mathrm{Binom}\left(\underline{E}_1^{k-r+1}(t)+t, \frac{1}{N(t)-k+r}\right).
\end{align*}
The probabilities~\eqref{eq:psl} and \eqref{eq:p2l} can be lower bounded by replacing the original Markov chain by the dominated Markov chain.
Indeed, the dominating Markov chain~\eqref{eq:MarkovU} is very close to the dominated Markov chain~\eqref{eq:MarkovL} for $t=O(\ell(n))$.

\begin{theorem}[Moment generating function of $\underline{E}_1^{k-r+1}(t)$]
\label{thm:LG}
Assume $m(n)=\mu\frac{n^{r-1}}{\ell(n)^{r-2}}$ for arbitrary constant $\mu$ and $\ell(n)\in \omega(1)\cap o(n)$.
Then, for any constants $\tau>0$ and $\lambda$, it holds
$\mathbb{E}[\exp\{\lambda\underline{E}_1^{k-r+1}(\tau \ell(n))\}]=
\exp\{ \varphi_{k,r}(\mu,\lambda,\tau) \ell(n)\allowbreak+ O(\max\{1,\ell(n)^2/n\})\}$
where $\varphi_{k,r}(\mu,\lambda,\tau)$ is defined in~\eqref{eq:Ek}.
\end{theorem}
The proof is omitted since it is straightforward from the proof of Theorem~\ref{thm:UG}.
From Theorem~\ref{thm:LG}, if $\mu>\mu_\mathrm{c}(k,r)$,
it holds
\begin{align}
&\Pr\left(\bigcup_{t=0}^{\lfloor\tau \ell(n)\rfloor-1}\underline{E}_1^{k-r+1}(t)\le 0\right) \le
\sum_{t=0}^{\lfloor\tau \ell(n)\rfloor-1}\Pr\left(\underline{E}_1^{k-r+1}(t)\le 0\right)\nonumber\\
&\qquad\le \sum_{t=0}^{\lfloor\tau \ell(n)\rfloor-1} \inf_{\lambda<0}\mathbb{E}\left[\exp\left\{\lambda\underline{E}_1^{k-r+1}(t)\right\}\right]\nonumber\\
&\qquad\le \exp\left\{\sup_{\tau'>0}\inf_{\lambda<0}\left\{\varphi_{k,r}(\mu,\lambda,\tau')\right\}\ell(n)+ O(\max\{\ell(n)^2/n,\log \ell(n)\})\right\}.
\label{eq:supinf}
\end{align}
Note that the above upper bound is independent of $\tau$.
In the same way, one can show that if $\mu>\mu_\mathrm{c}(k,r)$, for any $\eta >0$ and any $c>0$, there is sufficiently large $\tau>0$, such that
\begin{align*}
&\Pr\left(\underline{E}_1^{k-r+1}(\lfloor\tau \ell(n)\rfloor)< \eta \ell(n)\right)\le \exp\{-c \ell(n)\}.
\end{align*}

Similarly to Theorem~\ref{thm:LG}, asymptotic analysis of the moment generating function for $\underline{C}_{k-r+2}(t)$
 is obtained for $t=\Theta(\ell(n))$.
\begin{theorem}[Moment generating function of $\underline{C}_j(t)$]\label{thm:LG2}
Assume $m(n)=\mu\frac{n^{r-1}}{\ell(n)^{r-2}}$ for arbitrary constant $\mu$ and $\ell(n)\in \omega(1)\cap o(n)$.
Then, for any $j=k-r+2,\dotsc,k$, for any constants $\tau>0$ and $\lambda_j$,
\begin{align*}
\mathbb{E}[\exp\{\lambda_j\underline{C}_j(\lfloor\tau \ell(n)\rfloor)\}]&=
\exp\Biggl\{ \varphi^{(j)}_{k,r}(\mu,\lambda_j,\tau) \frac{n^{j-k+r-1}}{\ell(n)^{j-k+r-2}}
+ O\left(\frac{n^{j-k+r-1}}{\ell(n)^{j-k+r-1}}\max\left\{1,\frac{\ell(n)^2}{n}\right\}\right)\Biggr\}
\end{align*}
where
\begin{equation*}
\varphi^{(j)}_{k,r}(\mu,\lambda,\tau):= \mu \left(\exp\{\lambda\}-1\right)\binom{k}{k-j}(1+\tau)^{k-j}.
\end{equation*}
\end{theorem}
The proof of this theorem is also omitted since it is straightforward from the proof of Theorem~\ref{thm:UG}.
From Theorem~\ref{thm:LG2}, for any $\tau>1/(r-2)$, it holds
\begin{align*}
&
\Pr\left(\underline{C}_{k-r+2}(\lfloor\tau \ell(n)\rfloor)\le([(k-r+2)(k-r+1)]^{-1}+\epsilon)n\right)\\
&\le
\frac{\mathbb{E}\left[\exp\left\{\lambda_{k-r+2} \underline{C}_{k-r+2}(\lfloor\tau \ell(n)\rfloor)\right\}\right]}{\exp\{\lambda_{k-r+2}([(k-r+2)(k-r+1)]^{-1}+\epsilon)n\}}\\
&\le
\frac{\mathbb{E}\left[\exp\left\{\lambda_{k-r+2} \underline{C}_{k-r+2}(\lfloor \ell(n)/(r-2)\rfloor)\right\}\right]}{\exp\{\lambda_{k-r+2}([(k-r+2)(k-r+1)]^{-1}+\epsilon)n\}}\\
&=
\exp\biggl\{\mu \left(\exp\{\lambda_{k-r+2}\}-1\right)\frac{1}{(k-r+2)(k-r+1)}\binom{k}{r}\frac{r(r-1)^{r-1}}{(r-2)^{r-2}} n\\
&\qquad - \lambda_{k-r+2}([(k-r+2)(k-r+1)]^{-1}+\epsilon)n\biggr\}\\
&=
\exp\biggl\{[(k-r+2)(k-r+1)]^{-1}\frac{\mu}{\mu_\mathrm{c}(k,r)} \left(\exp\{\lambda_{k-r+2}\}-1\right) n\\
&\qquad - \lambda_{k-r+2}([(k-r+2)(k-r+1)]^{-1}+\epsilon)n\biggr\}
\end{align*}
for any $\lambda_{k-r+2}\le 0$.
Hence, if $\mu>\mu_\mathrm{c}(k,r)$, for sufficiently small $\epsilon>0$, there is $\delta>0$ such that
\begin{align*}
\Pr\left(\underline{C}_{k-r+2}(\ell(n)/(r-2))\le([(k-r+2)(k-r+1)]^{-1}+\epsilon)n\right)&\le
\exp\{-\delta n\}.
\end{align*}
The probability that the peeling algorithm does not remove linearly many vertex nodes is dominated by~\eqref{eq:supinf}.
%
By calculation of the saddle point, one obtains~\eqref{eq:rho} and \eqref{eq:tau}.

\section{Evolution of the moment generating function}\label{sec:gen}
In this section, the proof of Theorem~\ref{thm:UG} is shown.
The moment generating function for $[(\overline{E}_1^{k-r+1}(t)+t)/(k-r+1), \overline{C}_{k-r+2}(t),\dotsc,\overline{C}_k(t)]$ is defined as
\begin{align*}
&\overline{f}_t(\lambda_{k-r+1},\dotsc,\lambda_k)\\
&\,:=
\mathbb{E}\left[
\exp\left\{\lambda_{k-r+1}(\overline{E}_1^{k-r+1}(t)+t)/(k-r+1)+\lambda_{k-r+2}\overline{C}_{k-r+2}(t)+
\dotsb+\lambda_k\overline{C}_k(t)\right\}\right].
\end{align*}
From~\eqref{eq:MarkovU}, one obtains a recursive formula
\begin{align*}
&\overline{f}_{t+1}(\lambda_{k-r+1},\dotsc,\lambda_k)\\
&=\mathbb{E}\Biggl[\exp\left\{\lambda_{k-r+1}(\overline{E}_1^{k-r+1}(t)+t)/(k-r+1)+\lambda_{k-r+2}\overline{C}_{k-r+2}(t)+\dotsb+\lambda_k\overline{C}_k(t)\right\}\\
&\qquad\cdot
\exp\left\{\lambda_{k-r+1}\overline{R}_{k-r+2}(t) + \lambda_{k-r+2}\overline{R}_{k-r+3}(t)+\dotsb+\lambda_{k-1}\overline{R}_k(t)\right\}\Biggr]\\
&=\mathbb{E}\Biggl[\exp\left\{\lambda_{k-r+1}(\overline{E}_1^{k-r+1}(t)+t)/(k-r+1)+\lambda_{k-r+2}\overline{C}_{k-r+2}(t)+\dotsb+\lambda_k\overline{C}_k(t)\right\}\\
&\quad\cdot\prod_{j=k-r+2}^{k}\left(1-\frac{j}{N(t)}+\frac{j}{N(t)}\exp\{\lambda_{j-1}\}\right)^{\overline{C}_j(t)}\Biggr]\\
&= \overline{f}_t(\lambda_{k-r+1},\lambda'_{k-r+2},\dotsc,\lambda'_k)
\end{align*}
where
\begin{align*}
\lambda'_j &:= \lambda_j + \log\left(1-\frac{j}{N(t)}+\frac{j}{N(t)}\exp\{\lambda_{j-1}\}\right)
\end{align*}
for $j=k-r+2,k-r+3,\dotsc,k$.
Let $\lambda_{k-r+1}^{(s)}:= \lambda_{k-r+1}$ for $s=1,2,\dotsc,t$.
For $j=k-r+2,k-r+3,\dotsc,k$,
let
$\lambda^{(0)}_j:= 0$ and
\begin{align*}
\lambda^{(s)}_j &:= \lambda^{(s-1)}_j + \log\left(1-\frac{j}{N(t-s+1)}+\frac{j}{N(t-s+1)}\exp\{\lambda^{(s-1)}_{j-1}\}\right)
\end{align*}
for $s=1,2,\dotsc,t$.
Then, it holds
\begin{align*}
\mathbb{E}[\exp\{\lambda_{k-r+1} (\overline{E}_1^{k-r+1}(t)+t)/(k-r+1)\}]
&=\overline{f}_{t}(\lambda_{k-r+1},0,\dotsc,0)\\
&= \overline{f}_0\bigl(\lambda_{k-r+1}^{(t)},\lambda_{k-r+2}^{(t)},\dotsc,\lambda_k^{(t)}\bigr).
\end{align*}

\begin{lemma}
For $t=O(\ell(n))$ and $\ell(n)=o(n)$, it holds
\begin{align*}
\exp\{\lambda_j^{(t)}\} &=
 1 + \binom{j}{k-r+1}\frac{t^{j-k+r-1}}{n^{j-k+r-1}}\left(\exp\{\lambda_{k-r+1}\}-1\right)\\
&\quad + O\left(\frac{\ell(n)^{j-k+r-2}}{n^{j-k+r-1}}\max\left\{1,\frac{\ell(n)^2}{n}\right\}\right)
\end{align*}
for $j=k-r+1,k-r+2,\dotsc,k$.
\begin{proof}
The lemma is shown by induction on $j$.
The lemma obviously holds for $j=k-r+1$.
Assume the lemma holds for $j=j_0-1\ge k-r+1$, then
\begin{align*}
&\lambda_{j_0}^{(t)} = \sum_{s=0}^{t-1}\log\left(1-\frac{j_0}{N(t-s)}+\frac{j_0}{N(t-s)}\exp\{\lambda_{j_0-1}^{(s)}\}\right)\\
&= \sum_{s=0}^{t-1}\frac{j_0}{N(t-s)}\left(\exp\{\lambda_{j_0-1}^{(s)}\}-1\right) + O\left(\frac{\ell(n)^{2(j_0-k+r-1)-1}}{n^{2(j_0-k+r-1)}}\right)\\
&= \sum_{s=0}^{t-1}\frac{j_0}{n}\left(\exp\{\lambda_{j_0-1}^{(s)}\}-1\right) + O\left(\frac{\ell(n)^{j_0-k+r}}{n^{j_0-k+r}}\right)\\
&= \sum_{s=0}^{t-1}\frac{j_0}{n}\binom{j_0-1}{k-r+1}\frac{s^{j_0-k+r-2}}{n^{j_0-k+r-2}}(\exp\{\lambda_{k-r+1}\}-1) + O\left(\frac{\ell(n)^{j_0-k+r-2}}{n^{j_0-k+r-1}}\max\left\{1,\frac{\ell(n)^2}{n}\right\}\right)\\
&= \binom{j_0}{k-r+1}\frac{t^{j_0-k+r-1}}{n^{j_0-k+r-1}}(\exp\{\lambda_{k-r+1}\}-1) + O\left(\frac{\ell(n)^{j_0-k+r-2}}{n^{j_0-k+r-1}}\max\left\{1,\frac{\ell(n)^2}{n}\right\}\right).
\qedhere
\end{align*}
\end{proof}
\end{lemma}

Since the random variables $[C_0(0),\dotsc,C_k(0)]$ at the initial step 
obey the multinomial distribution $\mathrm{Multinom}(m(n),\allowbreak p_0(n),\dotsc,p_k(n))$,
it holds for $t=\lfloor\tau \ell(n)\rfloor$ and $m(n)=\mu\frac{n^{r-1}}{\ell(n)^{r-2}}$ that
\begin{align*}
&\overline{f}_0(\lambda_{k-r+1}^{(t)},\lambda_{k-r+2}^{(t)},\dotsc,\lambda_k^{(t)})\\
&=\left(p_0(n) + \sum_{j=1}^{k-r+1}p_j(n)\exp\left\{\frac{j}{k-r+1}\lambda_{k-r+1}\right\}
+ \sum_{j=k-r+2}^k p_j(n)\exp\left\{\lambda_j^{(t)}\right\}\right)^{m(n)}\\
&=\Biggl(1+\sum_{j=k-r+1}^k p_j(n)\binom{j}{k-r+1}\frac{t^{j-k+r-1}}{n^{j-k+r-1}}(\exp\{\lambda_{k-r+1}\}-1)\\
&\qquad + O\left(\frac{\ell(n)^{r-2}}{n^{r-1}}\max\left\{1,\frac{\ell(n)^2}{n}\right\}\right)\Biggr)^{m(n)}\\
&=\Biggl(1+\frac{\ell(n)^{r-1}}{n^{r-1}} \left(\exp\{\lambda_{k-r+1}\}-1\right)\sum_{j=k-r+1}^k \binom{k}{j}\binom{j}{k-r+1}\tau^{j-k+r-1}\\
&\qquad + O\left(\frac{\ell(n)^{r-2}}{n^{r-1}}\max\left\{1,\frac{\ell(n)^2}{n}\right\}\right)\Biggr)^{m(n)}\\
&=\exp\left\{\ell(n)\mu \left(\exp\{\lambda_{k-r+1}\}-1\right)\binom{k}{r-1}(1+\tau)^{r-1}+ O\left(\max\left\{1,\frac{\ell(n)^2}{n}\right\}\right)\right\}.
\end{align*}
From
\begin{align*}
&\mathbb{E}[\exp\{\lambda (\overline{E}_1^{k-r+1}(\lfloor\tau \ell(n)\rfloor)+\lfloor\tau \ell(n)\rfloor)/(k-r+1)\}]\\
&\quad =\exp\left\{\ell(n)\mu \left(\exp\{\lambda\}-1\right)\binom{k}{r-1}(1+\tau)^{r-1}+ O\left(\max\left\{1,\frac{\ell(n)^2}{n}\right\}\right)\right\}
\end{align*}
one obtains
\begin{align*}
\mathbb{E}[\exp\{\lambda \overline{E}_1^{k-r+1}(\lfloor\tau \ell(n)\rfloor)\}]
&=\exp\biggl\{\left[\mu \left(\exp\{(k-r+1)\lambda\}-1\right)\binom{k}{r-1}(1+\tau)^{r-1}-\lambda\tau\right]\\
&\qquad\times \ell(n) + O\left(\max\left\{1,\frac{\ell(n)^2}{n}\right\}\right)\biggr\}.
\end{align*}

\bibliographystyle{IEEEtran}
\bibliography{IEEEabrv,bibliography}

\begin{thebibliography}{10}
\providecommand{\url}[1]{#1}
\csname url@samestyle\endcsname
\providecommand{\newblock}{\relax}
\providecommand{\bibinfo}[2]{#2}
\providecommand{\BIBentrySTDinterwordspacing}{\spaceskip=0pt\relax}
\providecommand{\BIBentryALTinterwordstretchfactor}{4}
\providecommand{\BIBentryALTinterwordspacing}{\spaceskip=\fontdimen2\font plus
\BIBentryALTinterwordstretchfactor\fontdimen3\font minus
  \fontdimen4\font\relax}
\providecommand{\BIBforeignlanguage}[2]{{%
\expandafter\ifx\csname l@#1\endcsname\relax
\typeout{** WARNING: IEEEtran.bst: No hyphenation pattern has been}%
\typeout{** loaded for the language `#1'. Using the pattern for}%
\typeout{** the default language instead.}%
\else
\language=\csname l@#1\endcsname
\fi
#2}}
\providecommand{\BIBdecl}{\relax}
\BIBdecl

\bibitem{Luby:1997:PLC:258533.258573}
M.~G. Luby, M.~Mitzenmacher, M.~A. Shokrollahi, D.~A. Spielman, and V.~Stemann,
  ``Practical loss-resilient codes,'' in \emph{Proceedings of the Twenty-ninth
  Annual ACM Symposium on Theory of Computing}, ser. STOC '97.\hskip 1em plus
  0.5em minus 0.4em\relax New York, NY, USA: ACM, 1997, pp. 150--159.

\bibitem{ibrahimi2011set}
\BIBentryALTinterwordspacing
M.~Ibrahimi, Y.~Kanoria, M.~Kraning, and A.~Montanari, ``The set of solutions
  of random {XORSAT} formulae,'' in \emph{Proceedings of the Twenty-third
  Annual ACM-SIAM Symposium on Discrete Algorithms}, ser. SODA '12.\hskip 1em
  plus 0.5em minus 0.4em\relax SIAM, 2012, pp. 760--779. [Online]. Available:
  \url{http://dl.acm.org/citation.cfm?id=2095116.2095178}
\BIBentrySTDinterwordspacing

\bibitem{dietzfelbinger2010tight}
M.~Dietzfelbinger, A.~Goerdt, M.~Mitzenmacher, A.~Montanari, R.~Pagh, and
  M.~Rink, ``Tight thresholds for cuckoo hashing via {XORSAT},'' in
  \emph{Automata, Languages and Programming}.\hskip 1em plus 0.5em minus
  0.4em\relax Springer, 2010, pp. 213--225.

\bibitem{mitzenmacher2013simple}
M.~Mitzenmacher and R.~Pagh, ``Simple multi-party set reconciliation,''
  http://arxiv.org/abs/1311.2037v1, 2013.

\bibitem{watanabe2013mp}
O.~Watanabe, ``\BIBforeignlanguage{English}{Message passing algorithms for
  {MLS-3LIN} problem},'' \emph{\BIBforeignlanguage{English}{Algorithmica}},
  vol.~66, no.~4, pp. 848--868, 2013.

\bibitem{coja2012propagation}
A.~Coja-Oghlan, M.~Onsj{\"o}, and O.~Watanabe, ``Propagation connectivity of
  random hypergraphs,'' \emph{The Electronic Journal of Combinatorics},
  vol.~19, no.~1, p. P17, 2012.

\bibitem{connamacher2012exact}
H.~Connamacher, ``Exact thresholds for {DPLL} on random {XOR-SAT} and
  {NP}-complete extensions of {XOR-SAT},'' \emph{Theoretical Computer Science},
  vol. 421, pp. 25--55, 2012.

\bibitem{odonnell2014goldreich}
R.~O\rq~Donnell and D.~Witmer, ``Goldreich\rq s {PRG}: Evidence for
  near-optimal polynomial stretch,'' in \emph{Proceedings of the 2014 IEEE
  Conference on Computational Complexity (CCC)}, 2014.

\bibitem{mct}
T.~Richardson and R.~Urbanke, \emph{Modern Coding Theory}.\hskip 1em plus 0.5em
  minus 0.4em\relax Cambridge University Press, 2008.

\bibitem{wormald1995differential}
N.~C. Wormald, ``Differential equations for random processes and random
  graphs,'' \emph{The annals of applied probability}, pp. 1217--1235, 1995.

\bibitem{Achlioptas2001159}
D.~Achlioptas, ``Lower bounds for random 3-{SAT} via differential equations,''
  \emph{Theoretical Computer Science}, vol. 265, no.~1, pp. 159--185, 2001.

\bibitem{schmidt1985component}
J.~Schmidt-Pruzan and E.~Shamir, ``\BIBforeignlanguage{English}{Component
  structure in the evolution of random hypergraphs},''
  \emph{\BIBforeignlanguage{English}{Combinatorica}}, vol.~5, no.~1, pp.
  81--94, 1985.

\bibitem{di2002finite}
C.~Di, D.~Proietti, I.~E. Telatar, T.~J. Richardson, and R.~L. Urbanke,
  ``Finite-length analysis of low-density parity-check codes on the binary
  erasure channel,'' \emph{{IEEE} Trans. Inf. Theory}, vol.~48, no.~6, pp.
  1570--1579, 2002.

\bibitem{RSA:RSA20160}
A.~Coja-Oghlan, C.~Moore, and V.~Sanwalani, ``Counting connected graphs and
  hypergraphs via the probabilistic method,'' \emph{Random Structures and
  Algorithms}, vol.~31, no.~3, pp. 288--329, 2007.

\bibitem{behrisch2007local}
M.~Behrisch, A.~Coja-Oghlan, and M.~Kang, ``Local limit theorems for the giant
  component of random hypergraphs,'' in \emph{Approximation, Randomization, and
  Combinatorial Optimization. Algorithms and Techniques}.\hskip 1em plus 0.5em
  minus 0.4em\relax Springer, 2007, pp. 341--352.

\end{thebibliography}





\appendix
\section{Proof of Lemma~\ref{lem:linear}}\label{apdx:linear}
Let $S(l)$ be a random variable corresponding to the number of $d$-stopping sets of size $l$ in the randomly generated hypergraph.
Then, the probability that the randomly generated hypergraph includes at least one $\alpha$-linear $d$-stopping set is
upper bounded by using Markov's inequality as
\begin{equation*}
\Pr\left(\sum_{l=\lceil \alpha n\rceil +1}^{\lfloor (1-\alpha) n\rfloor}S(l)\ge 1\right)
\le
\sum_{l=\lceil \alpha n\rceil +1}^{\lfloor (1-\alpha) n\rfloor} \mathbb{E}[S(l)].
\end{equation*}
The expected number of $d$-stopping sets of size $l$ is equal to
\begin{equation*}
\mathbb{E}[S(l)]
=
\binom{n}{l}
\left(\sum_{s=0, d,d+1,\dotsc,k} \frac{\binom{l}{s}\binom{n-l}{k-s}}{\binom{n}{k}}\right)^{m(n)}.
\end{equation*}
Especially for $d=2$, it holds
\begin{equation*}
\mathbb{E}[S(l)]
=
\binom{n}{l}
\left(1-\frac{l\binom{n-l}{k-1}}{\binom{n}{k}}\right)^{m(n)}.
\end{equation*}
When $m(n)=\gamma n$ for some constant $\gamma>0$, it holds
\begin{equation*}
\frac1n\log \mathbb{E}[S(\delta n)]
=
h(\delta)
+
\gamma \log \left(1-k \delta(1-\delta)^{k-1}\right)
+ o(1)
\end{equation*}
for any $\delta\in(0,1)$ where $h$ denotes the binary entropy function.
Hence, for any fixed $\alpha\in(0,1/2)$, there is a constant $\gamma_\alpha$ such that
\begin{equation*}
h(\delta)+
\gamma_\alpha \log \left(1-k \delta(1-\delta)^{k-1}\right)
\le -1
\end{equation*}
for any $\delta\in[\alpha,1-\alpha]$.
Hence,
\begin{equation*}
\sum_{l=\lceil \alpha n\rceil +1}^{\lfloor (1-\alpha) n\rfloor} \mathbb{E}[S(l)]
\le n\exp\{-n+o(n)\}
\end{equation*}
when $m(n)=\gamma_\alpha n$.

\section{Proof of Lemma~\ref{lem:small}}\label{apdx:small}
From an inequality
\begin{equation*}
\log\left(1-\frac{l\binom{n-l}{k-1}}{\binom{n}{k}}\right)
\le
-\frac{l\binom{n-l}{k-1}}{\binom{n}{k}}
\end{equation*}
one obtains for $m(n)=\mu n\log n$ that
\begin{align*}
\sum_{l=1}^{\delta n}\mathbb{E}[S(l)]
&\le
\sum_{l=1}^{\delta n}\binom{n}{l}\exp\left\{-m(n)\frac{l\binom{n-\delta n}{k-1}}{\binom{n}{k}}\right\}
\le\left(1+\exp\left\{-m(n)\frac{\binom{n-\delta n}{k-1}}{\binom{n}{k}}\right\}\right)^n-1\\
&=
\left(1+n^{-\mu k(1-\delta)^{k-1} + o(1)}\right)^n-1
\end{align*}
for any $\delta\in(0,\alpha)$.
Let $\delta_\mu := 1-1/(\mu k)^{1/(k-1)}$.
For any $\mu>1/k$ and any $\delta \in (0,\delta_\mu)$, it holds $\mu k(1-\delta)^{k-1}>1$,
i.e.,
\begin{equation*}
\left(1+n^{-\mu k(1-\delta)^{k-1} + o(1)}\right)^n-1
=O\left(n^{1-\mu k(1-\delta)^{k-1}}\right).
\qedhere
\end{equation*}

\section{Analyses of stopping sets for $r=2$}\label{apx:r2}
In this section, the existence of $\alpha$-large $k$-stopping set is analyzed.
Lemma~\ref{lem:2large} in this section is used in Section~\ref{subsec:lower}.
For $\alpha$-large $k$-stopping set, that corresponds to the case $r=2$, the threshold is obtained as follows.
\begin{lemma}\label{lem:2large}
For any $\mu>(k(k-1))^{-1}$, there exists $\alpha\in(0,1)$
such that
$G_k(n, \mu n, \ell(n))$
does not have $k$-stopping set of size greater than $\alpha n$ with probability exponentially close to 1 with respect to $\ell(n)$.
\end{lemma}
\begin{proof}
From the theory of random hypergraphs, if $m(n)=\mu n$ where $\mu> (k(k-1))^{-1}$, then the random hypergraph including $n$ vertices and $m(n)$ hyperedges has a giant component,
which is a connected component of size proportional to $n$,
 with probability approaching to 1 exponentially fast as $n\to\infty$~\cite{schmidt1985component}, \cite{behrisch2007local}.
It is also shown in~\cite{RSA:RSA20160} that the size of giant component is concentrated around $(1-\rho)n$ where $\rho\in(0,1)$ is the unique solution of
\begin{equation*}
\rho=\exp\{\mu k (\rho^{k-1}-1)\}.
\end{equation*}
Hence, the probability that the size of giant component is greater than $(1-\rho-\delta)n$ tends to 1 exponentially fast with respect to $n$
for any $\delta>0$.
If at least one of the $\ell(n)$ vertices are included in the giant component, the giant component is removed by the $k$-peeling algorithm.
In that case, the size of the largest stopping set is at most $(\rho+\delta)n$.
The probability that all of the $\ell(n)$ removed vertex nodes are not included in the giant component is at most
$(\rho+\delta)^{\ell(n)}$.
\end{proof}
%
%
The converse of Lemma~\ref{lem:2large} is also obtained as follows.
\begin{lemma}\label{lem:2ML}
For any $\mu<(k(k-1))^{-1}$,
$G_k(n, \mu n, \ell(n))$ has
$k$-stopping set of size larger than $n-(1+\tau) \ell(n)$ with high probability
for any $\tau$ strictly larger than
\begin{equation*}
\frac{k(k-1)\mu}{1-k(k-1)\mu}.
\end{equation*}
\end{lemma}
\begin{proof}
The proof is almost same as the proof of Theorem~\ref{thm:ML}.
In~\eqref{eq:sup},
the supremum is taken at $\tau\to+\infty$ when $r=2$, and hence the condition~\eqref{eq:sup} is equivalent to $\mu < [k(k-1)]^{-1}$.
\end{proof}

\end{document}